\documentclass[11pt]{article}
\usepackage{fullpage}           
\usepackage{amsfonts}           
\usepackage{amsthm}             
\usepackage{amsmath}
\usepackage{amssymb}
\usepackage{xspace}             
\usepackage{graphicx}            
\usepackage{adjustbox}          

\usepackage{multicol}            
\usepackage{url}        
\usepackage{thmtools}
\usepackage{thm-restate}
\usepackage{enumerate}  
\usepackage{subcaption}
\usepackage[usenames]{xcolor} 

\usepackage{hyperref}
\usepackage{mathabx}

\usepackage{algorithm}
\usepackage[noend]{algpseudocode}

\usepackage{stackengine}
\stackMath
\newcommand\tsup[2][2]{%
 \def\useanchorwidth{T}%
  \ifnum#1>1%
    \stackon[-1.3ex]{\tsup[\numexpr#1-1\relax]{#2}}{\mathchar"307E}%
  \else%
    \stackon[-1ex]{#2}{\mathchar"307E}%
  \fi%
}

\usepackage{nccmath}
\newtheorem{theorem}{Theorem}
\newtheorem{lemma}[theorem]{Lemma}

\newtheorem{definition}[theorem]{Definition}
\newtheorem{obs}[theorem]{Observation}

\newtheorem{claim}[theorem]{Claim}
\newtheorem{example}[theorem]{Example}

\newcommand{\Oh}{O}
\newcommand{\tOh}{\widetilde{O}}
\newcommand{\Os}{O^*}

\DeclareMathOperator{\polylog}{polylog}
\DeclareMathOperator{\poly}{poly}
\DeclareMathOperator{\LL}{Log}

\newcommand{\cS}{{\cal S}}
\newcommand{\tI}{{\widetilde I}}
\newcommand{\tB}{{\widetilde B}}
\newcommand{\tT}{{\widetilde T}}
\newcommand{\tX}{{\widetilde X}}
\newcommand{\tY}{{\widetilde Y}}
\newcommand{\dbtilde}[1]{\widetilde{\raisebox{0pt}[0.85\height]{$\widetilde{#1}$}}}
\newcommand{\ttI}{{\dbtilde I}}
\newcommand{\ttB}{{\dbtilde B}}
\newcommand{\ttT}{{\dbtilde T}}
\newcommand{\ttX}{{\dbtilde X}}
\newcommand{\ttY}{{\dbtilde Y}}
\newcommand{\NN}{\mathbb{N}}
\newcommand{\ZZ}{\mathbb{Z}}
\newcommand{\RR}{\mathbb{R}}
\newcommand{\RRp}{\mathbb{R}_{>0}}

\newcommand{\Sum}{\Sigma}
\newcommand{\eps}{\varepsilon}
\newcommand{\OPT}{\textup{OPT}}
\newcommand{\OPTL}{\textup{OPT}_{\textup{L}}}
\newcommand{\SSR}{$\textup{SSR}$\xspace}
\newcommand{\SSRL}{$\textup{SSR}_{\textup{L}}$\xspace}
\newcommand{\arr}[2]{[#1 \,..\, #2]}
\newcommand{\opU}{{\cal U}}

\newcounter{sideremark}

\usepackage{setspace}

\title{Approximating Subset Sum Ratio faster than Subset Sum}
\date{}

\author{
  Karl Bringmann\thanks{Saarland University and Max-Planck-Institute for Informatics, Saarland Informatics Campus, Saarbr\"ucken, Germany.
   \texttt{bringmann@cs.uni-saarland.de}. This work is part of the project TIPEA that has received funding from the European Research Council (ERC) under the European Unions Horizon 2020 research and innovation programme (grant agreement No.\ 850979).
  }
}

\begin{document}

\maketitle

\begin{abstract}

Subset Sum Ratio is the following optimization problem: Given a set of $n$ positive numbers~$I$, find disjoint subsets $X,Y \subseteq I$ minimizing the ratio $\max\{\Sum(X)/\Sum(Y),\Sum(Y)/\Sum(X)\}$, where $\Sum(Z)$ denotes the sum of all elements of $Z$. 
Subset Sum Ratio is an optimization variant of the Equal Subset Sum problem. It was introduced by Woeginger and Yu in '92 and is known to admit an FPTAS [Bazgan, Santha, Tuza '98]. The best approximation schemes before this work had running time $O(n^4/\eps)$ [Melissinos, Pagourtzis~'18], $\tOh(n^{2.3}/\eps^{2.6})$ and $\tOh(n^2/\eps^3)$ [Alonistiotis et al.~'22].

In this work, we present an improved approximation scheme for Subset Sum Ratio running in time $O(n / \eps^{0.9386})$. Here we assume that the items are given in sorted order, otherwise we need an additional running time of $O(n \log n)$ for sorting. Our improved running time simultaneously improves the dependence on $n$ to linear and the dependence on $1/\eps$ to sublinear.

For comparison, the related Subset Sum problem admits an approximation scheme running in time $O(n/\eps)$ [Gens, Levner '79]. 
If one would achieve an approximation scheme with running time $\tOh(n / \eps^{0.99})$ for Subset Sum, then one would falsify the Strong Exponential Time Hypothesis [Abboud, Bringmann, Hermelin, Shabtay '19] as well as the Min-Plus-Convolution Hypothesis [Bringmann, Nakos '21]. We thus establish that Subset Sum Ratio admits faster approximation schemes than Subset Sum. This comes as a surprise, since at any point in time before this work the best known approximation scheme for Subset Sum Ratio had a worse running time than the best known approximation scheme for Subset Sum.
%
\end{abstract}

\section{Introduction}


For a set $Z$ of numbers we write $\Sum(Z)$ for the sum of all elements of $Z$. The Equal Subset Sum problem asks for a given set $I$ of $n$ numbers whether any two non-empty disjoint subsets $X,Y \subseteq I$ have equal sum $\Sum(X) = \Sum(Y)$. This is a classic NP-complete problem. A natural optimization variant of Equal Subset Sum is to find two sets with as close sum as possible; this is the Subset Sum Ratio problem. More precisely, in the Subset Sum Ratio problem we are given a set $I$ of $N$ positive numbers and the task is to find disjoint subsets $X,Y \subseteq I$ minimizing the ratio $R(X,Y) := \max\{\Sum(X)/\Sum(Y),\Sum(Y)/\Sum(X)\}$. Here we interpret $z/0 = \infty$ for any $z\ge 0$ to exclude that $X$ or $Y$ are the empty set.
Both Equal Subset Sum and Subset Sum Ratio are well-studied variants of the classic Subset Sum problem, see, e.g.,~\cite{AlonistiotisAMP22,BazganST02,
MelissinosP18,MuchaNPW19,Nanongkai13,WoegingerY92}. 


Subset Sum Ratio was introduced by Woeginger and Yu in '92~\cite{WoegingerY92}, where they proved NP-hardness and designed a 1.324-approximation algorithm.  The first FPTAS for Subset Sum Ratio was presented by Bazgan, Santha, and Tuza~\cite{BazganST02}. This FPTAS was simplified by Nanongkai~\cite{Nanongkai13}. Both papers did not specify their running time, but mentioned that their time complexities were ``quite high ($O(n^5)$ or more)''~\cite{Nanongkai13}. 
Melissinos and Pagourtzis~\cite{MelissinosP18} improved the running time to $O(n^4/\eps)$. Recently, Alonistiotis et al.~\cite{AlonistiotisAMP22} designed approximation schemes running in time $\tOh(n^{2.3}/\eps^{2.6})$ and $\tOh(n^2/\eps^3)$, thus achieving different tradeoffs compared to the algorithm by Melissinos and Pagourtzis.


In this paper, we present an improved FPTAS for Subset Sum Ratio.

\begin{theorem} \label{thm:main}
  Subset Sum Ratio has a $(1+\eps)$-approximation algorithm running in time $O(n/\eps^{0.9386})$. 
\end{theorem}

Here we assume that the input items are given in sorted order, otherwise we need an additional running time of $O(n \log n)$ to sort the input. 
Our machine model is either the Real RAM or the Word~RAM with floating point input, as we discuss in detail in Section~\ref{sec:preliminaries}.

Note that we improve all previous running times in terms of both $n$ and~$\eps$. Our algorithm is the first approximation scheme for Subset Sum Ratio with a linear dependence on~$n$, and also the first with a sublinear dependence on $1/\eps$. 


\paragraph{Subset Sum}
To put our result in context, consider the related Subset Sum problem. The decision version of Subset Sum asks for a given set $I$ of $n$ positive integers and a target $t$ whether any subset $X \subseteq I$ sums to $\Sum(X) = t$. The optimization version of Subset Sum instead asks to compute the number $\OPT = \max\{\Sum(X) \mid X \subseteq I, \, \Sum(X) \le t\}$. Naturally, this can be relaxed to a $(1-\eps)$-approximation, which computes a subset $Y \subseteq I$ with $(1-\eps) \OPT \le \Sum(Y) \le \OPT$. 
Subset Sum admits approximation schemes running in time\footnote{By $\tOh$ notation we hide polylogarithmic factors, that is, $\tOh(T) = \bigcup_{c \ge 0} O(T \log^c T)$.} $\tOh(\min\{n/\eps, n + 1/\eps^2\})$~\cite{gens1978approximation,GensL79,KellererMPS03,KellererPS97}; this is known for over 20 years, see also~\cite{BringmannN21} for further lower order improvements. Recent conditonal lower bounds show that we cannot expect much faster algorithms: (1)~An approximation scheme for Subset Sum with running time $\Oh((n+1/\eps)^{1.999})$ would contradict the Min-Plus-Convolution Hypothesis~\cite{BringmannN21}. Moreover, (2) if Subset Sum has an approximation scheme with running time $2^{o(n)} / \eps^{0.999}$ then by setting $\eps < 1/t$ we would obtain an exact algorithm running in pseudopolynomial time $2^{o(n)} t^{0.999}$, which would contradict the Strong Exponential Time Hypothesis~\cite{ABHS19}. Note that the running time $\tOh(n / \eps^{0.9386})$ that we obtain for approximating Subset Sum Ratio is of the form (1) $\Oh((n+1/\eps)^{1.999})$ and of the form (2)~$2^{o(n)} / \eps^{0.999}$. Hence, if we would obtain the same running time for approximating Subset Sum then we would falsify two standard hypotheses from fine-grained complexity theory. This gives strong evidence that \emph{Subset Sum Ratio can be approximated faster than Subset Sum}, which is surprising, since \emph{at any point in time prior to this paper the fastest known algorithms for Subset Sum Ratio were slower than the algorithms for Subset Sum}.

Our result is also surprising on an intuitive level: Since Subset Sum has a trivial search space of size $2^n$ (each item can be in the solution $X$ or not), while Subset Sum Ratio has a trivial search space of size $3^n$ (each item can be in $X$ or in $Y$ or in neither), it seems that Subset Sum is easier than Subset Sum Ratio. This intuition is supported by the state of the art for exact algorithms: Subset Sum is solved by meet-in-the-middle in time $O(2^{n/2}) = O(1.4143^n)$, while Equal Subset Sum (i.e., essentially the exact version of Subset Sum Ratio) is solved by meet-in-the-middle in time $O(3^{n/2})$, which has been recently improved to time $O(1.7088^n)$~\cite{MuchaNPW19} --- still the best known algorithm for Equal Subset Sum is much slower than the best known exact algorithm for Subset Sum, which supports the intuition that Subset Sum is easier than Subset Sum Ratio.
We show that this intuition is wrong for approximation schemes, as Subset Sum Ratio can be approximated faster than Subset Sum (conditional on a fine-grained hypothesis).

\paragraph{Fine-grained Complexity of Optimization Problems}
Our result is part of a recent effort of the fine-grained complexity community to obtain improved, and sometimes even best-possible approximation schemes for classic optimization problems such as Subset Sum, Knapsack, and Partition, see, e.g.~\cite{BringmannC22,BringmannN21,Chan18a,DengJM23,
Jin19,MuchaW019,WuChen22}. 
These three problems can all be solved in near-linear time $\tOh(n)$ when $\eps$ is constant; our result establishes that the same holds for Subset Sum Ratio. The optimal dependence on $\eps$ is essentially settled for Subset Sum, since it can be solved in time $\tOh(\min\{n/\eps,n+1/\eps^2\})$~\cite{gens1978approximation,GensL79,KellererMPS03,KellererPS97} but not in time $O((n+1/\eps)^{1.999})$ assuming the Min-Plus-Convolution Hypothesis~\cite{BringmannN21}. 
Since recently it is also essentially settled for Knapsack, since it can be solved in time $\tOh(n + 1/\eps^2)$~\cite{MaoArxiv23,ChenLMZArxiv23} but not in time $O((n+1/\eps)^{1.999})$ assuming the Min-Plus-Convolution Hypothesis~\cite{KunnemannPS17,CyganMWW19}.
For Partition, there is a gap between an upper bound of $\tOh(n + 1/\eps^{1.25})$~\cite{DengJM23} and a conditional lower bound ruling out $O(n + 1/\eps^{0.999})$~\cite{ABHS19}. 
Our result shows that the time complexity of Subset Sum Ratio is closer to Partition, since in the case $\eps = 1/n$ both Subset Sum Ratio and Partition can be solved in subquadratic time $O(n^{1.9386})$, while for Subset Sum and Knapsack this is ruled out by conditional lower bounds. However, we show that Subset Sum Ratio can be solved with sublinear dependence on $1/\eps$, in contrast to Partition, Subset Sum, and Knapsack which all have a linear lower bound in $1/\eps$. Thus, we establish an interesting new running time behaviour of optimization problems.
Determining the optimal dependence on $\eps$, even conditionally, remains an open problem for Knapsack, Partition, and Subset Sum Ratio.

\paragraph{Techniques} 
We contribute several novel ideas for solving Subset Sum Ratio. First, in Section~\ref{sec:simple} we develop a linear-time reduction that replaces one instance of size $n$ by $n$ instances of size $\polylog(1/\eps)$. Solving these $n$ instances independently naturally leads to a running time that depends linearly on $n$. For example, if we solve each of these $n$ instances of size $\polylog(1/\eps)$ by the known $O(n^4/\eps)$-time algorithm~\cite{MelissinosP18} then we obtain an algorithm running in total time $O(n/\eps \polylog(1/\eps))$, which already significantly improves the state of the art. 
We expect future work on Subset Sum Ratio to make frequent use this reduction. 

After this reduction, in Section~\ref{sec:improved} the remaining goal is to solve Subset Sum Ratio in time $O(\poly(n)/\eps^{0.9386})$. The specific polynomial dependence on $n$ no longer matters at this point. We combine a delicate pigeonhole argument with ideas from rounding, sumsets, and computational geometry in an intricate way to achieve our goal. Along the way, we solve instances with optimal value more than $\sqrt{2}$ to optimality. See Section~\ref{sec:overview} for a more detailed overview of this part.

\paragraph{Open Problems}
Our running time of $O(n/\eps^{0.9386})$ is unlikely to be optimal. We leave the following open problems: (1) Is a running time of the form $\tOh(n + 1/\poly(\eps))$ possible? (2) Can one show a (conditional) lower bound ruling out time $O(n/\eps^c)$ for some constant $c > 0$? (3) Determine the optimal constant $c \ge 0$ such that Subset Sum Ratio is in time $O(n/\eps^c)$.

\section{Preliminaries} \label{sec:preliminaries}


For a set $X$ we write $\Sum(X) := \sum_{x \in X} x$. For two sets $X,Y$ their ratio is defined as 
\[ R(X,Y) := \max\{\Sum(X)/\Sum(Y), \Sum(Y)/\Sum(X)\}.\]
The optimal ratio of a set $I \subset \RRp$ is 
\[ \OPT(I) := \min\{ R(X,Y) \mid \textup{disjoint } X,Y \subseteq I \}. \]
Throughout the paper we use the convention $x/0 = \infty$ for any $x \in \RR$ with $x \ge 0$. This ensures that the optimum ratio is attained by non-empty $X,Y$ (unless $|I| \le 1$).
With this notation, we can define the Subset Sum Ratio problem.

\begin{definition}[Subset Sum Ratio Problem (\SSR)] \label{def:ssr}
  Given a set $I \subset \RRp$ of size $n \ge 2$ and a number $\eps \in (0,1)$, compute disjoint subsets $X,Y \subseteq I$ with $R(X,Y) \le (1+\eps) \OPT(I)$. 
\end{definition}

Let $I$ be an instance of \SSR of size $n$. We write $I[1],\ldots,I[n]$ for its items, and we assume that these items are sorted as $I[1] \le \ldots \le I[n]$. We use the notation $I\arr{i}{j} := \{I[i],I[i+1],\ldots,I[j]\}$ for any integers $1 \le i \le j \le n$. We sometimes also use this notation for $i < 1$ or $j > n$, with the meaning $I\arr{i}{j} := I\arr{\max\{1,i\}}{\min\{n,j\}}$.


For sets $X,Y$ and real $\alpha$ we write $\alpha \cdot X := \{\alpha x \mid x \in X\}$, $X / \alpha := \alpha^{-1} \cdot X$, $X+Y := \{x+y \mid x \in X, y \in Y\}$, and $X-Y := X + ((-1) \cdot Y)$.

All logarithms in this paper are base 2, unless stated otherwise.

Since our algorithms use rounding, as intermediate objects we sometimes obtain multisets rather than sets. We therefore extend the above notation in the natural way to multisets. In particular, we say that $X,Y \subseteq I$ are disjoint if $X,Y$ are multisets such that for every $z \in \RRp$ the sum of the multiplicities of $z$ in $X$ and in $Y$ is at most the multiplicity of $z$ in $I$. It follows that whenever $I$ contains any number at least twice, then $\OPT(I) = 1$, by letting $X$ and $Y$ consist of this number. For this reason, multisets are a trivial boundary case (which nevertheless needs to be treated in our algorithms).

We sometimes use $\tOh$ notation to hide polylogarithmic factors, that is, $\tOh(T) = \bigcup_{c \ge 0} O(T \log^c T)$. In Section~\ref{sec:improved} we will use $\Os$ notation to hide $\poly(n)$ factors.

\paragraph{Machine Model} 
We assume a model of computation where standard arithmetic operations on input numbers, including rounding, can be performed in constant time. An example is the Real RAM model with floor function, as is used in many papers. 

Alternatively, our algorithms also work in a machine model with floating point arithmetic. Observe that, since we are only interested in a $(1+\eps)$-approximation for Subset Sum Ratio, we can safely change every input number by a factor $1+\eps$, in particular we can round every input number to a number of the form $(1 + \eps i) 2^j$, for integers $0 \le i \le 1/\eps$ and $j \in \ZZ$. In other words, without loss of generality we can assume that our input numbers come in floating point format with a $\log(1/\eps)$-bit mantissa. We assume that the exponent also is an $O(\log(n/\eps))$-bit integer, and we assume that the usual arithmetic operations on floating point numbers with $O(\log(n/\eps))$-bit mantissa and exponent take constant time. This is a realistic model of computation. Our algorithm can be analyzed in this model and has the same running time guarantee as on the Real RAM. We omit the error analysis that would be necessary to show that floating point approximation is sufficient inside our algorithm; instead we formally analyze our algorithm in the Real RAM model.

For comparison, the standard model of fine-grained complexity is the Word RAM, where each cell stores a $w$-bit integer for $w = O(\log n)$, and the usual arithmetic and logical operations on two cells take constant time. However, Subset Sum Ratio makes not much sense in this model, as it can be solved exactly in polynomial time. Indeed, the pigeonhole principle (see Observation~\ref{obs:pigeon} below) implies that any instance on more than $w$ many $w$-bit integers has optimal ratio $1$. Thus, we can ignore all but $w$ of the given $n$ input numbers, and find a ratio-1 solution among these $w$ numbers by brute force in time $O(3^w) = \poly(n)$. This issue is avoided by the floating point model since $O(\log(n/\eps))$-bit exponent and mantissa allow to store integers in the range $[1,2^{\poly(n/\eps)}]$, approximated up to a factor $1+\poly(\eps/n)$. The issue is avoided by the Real RAM model, since cells can store arbitrary integers. As discussed above, for simplicity we work in the Real RAM model in this paper.

\section{A Reduction to Few Items} \label{sec:simple}

In this section we present a reduction that replaces an \SSR instance of size $n$ by $n$ instances of size $\polylog(1/\eps)$. Combining this with known approximation schemes immediately yields a simple $\tOh(n/\eps)$-time algorithm for \SSR.

\subsection{Setting up the Relaxed Subproblem \SSRL}

Consider an \SSR instance $I$ of size $n$.
Suppose we guess the largest element appearing in an optimal solution $(X,Y)$. That is, we iterate over all choices of $\ell \in \{1,\ldots,n\}$ and for each choice we want to determine disjoint subsets $X,Y \subseteq I\arr{1}{\ell}$ such that $X$ or $Y$ contains $I[\ell]$ and $R(X,Y)$ is minimized. We introduce the following notation for this purpose:
\[ \OPTL(I) := \min\{ R(X,Y) \mid \textup{disjoint } X,Y \subseteq I,\, \max(I) \in X \cup Y \}. \]
By guessing the largest item $I[\ell]$ in an optimal solution and then restricting to sets containing $I[\ell]$, we observe that
\begin{align} \label{eq:optl} 
  \OPT(I) = \min_{1 \le \ell \le |I|} \OPTL(I\arr{1}{\ell}).
\end{align}
This naturally leads to the following subproblem.

\begin{definition}[Subset Sum Ratio containing Largest element (\SSRL)] \label{def:ssrl}
  Given a set $I \subset \RRp$ of size $n \ge 2$ and a number $\eps \in (0,1)$, compute disjoint subsets $X,Y \subseteq I$ with $R(X,Y) \le (1+\eps) \OPTL(I)$. 
\end{definition}

Comparing with the definition of \SSR (Definition~\ref{def:ssr}), the only difference is that we compare to $\OPTL(I)$ instead of $\OPT(I)$. A subtle and important detail of \SSRL is that we compare with the optimal solution containing the largest item (namely $\OPTL(I)$), but we do not enforce the computed solution $(X,Y)$ to contain the largest item. We will make use of this detail in Section~\ref{sec:improved}, where we design several algorithms for \SSRL, only some of which ensure picking the largest item.

\subsection{Reduction to \SSRL}

By equation (\ref{eq:optl}), in order to solve \SSR on instance $I$ it suffices to solve \SSRL on the $n$ instances $I\arr{1}{\ell}$ for all $1 \le \ell \le n$. This property was used in previous algorithms, see e.g.~\cite{Nanongkai13,MelissinosP18}. Note that this property alone is insufficient for our results, since already writing down all instances $I\arr{1}{\ell}$ would take total time $\Omega(n^2)$, while we will achieve a running time that has linear dependence on $n$. 

Our main new structural insight is that it suffices to consider the largest $\polylog 1/\eps$ items in $I\arr{1}{\ell}$. That is, instead of solving the \SSRL instances $I\arr{1}{\ell}$ we solve the \SSRL instances $I\arr{\ell-\polylog 1/\eps}{\ell}$ for each $1 \le \ell \le n$. More precisely, we prove the following analogue of equation (\ref{eq:optl}). 

\begin{lemma} \label{lem:reduction}
  For any $I \subset \RRp$ and any $\eps \in (0,1)$ we have
  \begin{align} \label{eq:reduction} 
    \OPT(I) \le \min_{1 \le \ell \le n} \OPTL(I\arr{\ell-\polylog 1/\eps}{\ell}) \le (1+\eps) \OPT(I).
  \end{align}
\end{lemma}
To describe the $\polylog 1/\eps$-term more precisely, we define for any $x \ge 1$ the number $\LL(x)$ to be the smallest integer with $2^{\LL(x)} > \LL(x)\cdot x + 1$. One can show that $\log x \le \LL(x) \le \log x + O(\log \log x)$. With this notation, the $\polylog 1/\eps$-term in the lemma can be set to $\LL^2(4/\eps) := (\LL(4/\eps))^2 = O(\log^2(1/\eps))$.


We will prove Lemma~\ref{lem:reduction} in Section~\ref{sec:proofreduction} below.
Let us first see that equation (\ref{eq:reduction}) yields a reduction from \SSR to \SSRL.

\begin{theorem}[Reduction]
  If \SSRL can be solved in time $T_{\textup{L}}(n,\eps)$, then \SSR can be solved in time 
  \[ O(n \cdot T_{\textup{L}}(\polylog 1/\eps,\, \eps/3)). \]
\end{theorem}
\begin{proof}
  Write $I_\ell := I\arr{\ell - \polylog 1/\eps}{\ell}$. By running an \SSRL algorithm on each instance $I_\ell$ we obtain disjoint subsets $X_\ell,Y_\ell \subseteq I_\ell$ with ratios $R(X_\ell,Y_\ell) \le (1+\eps) \cdot \OPTL(I_\ell)$. The best solution $(X,Y)$ among $(X_1,Y_1),\ldots,(X_n,Y_n)$ thus has a ratio of 
  \[ R(X,Y) = \min_{1 \le \ell \le n} R(X_\ell,Y_\ell) \le \min_{1 \le \ell \le n} (1+\eps) \cdot \OPTL(I_\ell) 
   \stackrel{(\ref{eq:reduction})}{\le} (1+\eps)^2 \OPT(I). \]
  Since for $\eps \in (0,1)$ we have $(1+\eps)^2 \le 1+3\eps$, by replacing $\eps$ by $\eps/3$ we obtain a solution to the \SSR instance~$I$. The total running time is $O(n \cdot T_{\textup{L}}(\polylog 1/\eps,\, \eps/3))$.
\end{proof}

Note that the above reduction naturally leads to algorithms for \SSR whose running time depends linearly on $n$, since the factor $T_{\textup{L}}(\polylog 1/\eps,\, \eps/3)$ is independent of $n$. 
Observe that any algorithm for \SSR also solves \SSRL (cf.\ Definitions~\ref{def:ssr} and~\ref{def:ssrl}). Thus, we can e.g.\ plug the known $\tOh(n^4/\eps)$-time algorithm for \SSR~\cite{MelissinosP18} into the above reduction. This yields an algorithm for \SSR with running time $O(\frac n \eps \polylog \frac 1 \eps)$. We remark that this already significantly improves the state of the art running time for \SSR.

In Section~\ref{sec:improved}, we will design an improved \SSRL algorithm with running time $O(\poly(n)/\eps^{0.93856})$. Plugging this into the above reduction yields time $O(n/\eps^{0.9386})$ (where we hide $\polylog(1/\eps)$ factors by rounding the exponent of $1/\eps$). This proves our main result Theorem~\ref{thm:main}. 

For now it remains to prove Lemma~\ref{lem:reduction}.


\subsection{Proof of Lemma~\ref{lem:reduction}} \label{sec:proofreduction}

Let us start by discussing the the main difficulty in proving Lemma~\ref{lem:reduction}:
The step from equation (\ref{eq:optl}) to (\ref{eq:reduction}) is non-trivial, because it can happen that for some values $\ell$ we have
\[ \OPTL(I\arr{\ell - \polylog 1/\eps}{\ell}) \gg \OPTL(I\arr{1}{\ell}). \] 
We illustrate this situation in the following example. 

\begin{example}
{\normalfont
Consider $I = \{1,2,\ldots,n-2,n-1,{n \choose 2}\}$ and $\ell = n$, so that $I\arr{1}{\ell} = I$. Then:
\begin{itemize}
  \item $\OPTL(I\arr{1}{\ell}) = 1$. \\ To see this, set $X = \{1,2,\ldots,n-2,n-1\}$ and $Y = \{{n \choose 2}\}$ and note that $X,Y$ are disjoint subsets of $I\arr{1}{\ell}$ and $Y$ contains $I[\ell]$, so $(X,Y)$ is a feasible solution for $\OPTL(I\arr{1}{\ell})$. Since $\Sum(X) = \sum_{i=1}^{n-1} i = {n \choose 2} = \Sum(Y)$ we have $R(X,Y) = 1$ and thus $\OPTL(I\arr{1}{\ell}) = 1$.
  \item $\OPTL(I\arr{\ell-\polylog 1/\eps}{\ell}) \gg 1$ assuming $n \gg \polylog 1/\eps$. \\ Indeed, note that $\polylog 1/\eps$ many items below $n$ sum to at most $n \cdot \polylog 1/\eps \ll {n \choose 2}$. Thus, for any disjoint subsets $X,Y$ with $\max(X \cup Y) = I[\ell]$ one of the sets sums to at least ${n \choose 2}$ and the other set sums to at most $n \cdot \polylog 1/\eps \ll {n \choose 2}$. It follows that their ratio is $\gg 1$. 
\end{itemize}
So indeed this is an example where for some value $\ell$ we have $\OPTL(I\arr{\ell - \polylog 1/\eps}{\ell}) \gg \OPTL(I\arr{1}{\ell})$.

%
Nevertheless, equation (\ref{eq:reduction}) still holds in this example (and in general, as we prove below). To see this, we shift our attention to $\ell' = n-1$, and observe that:
\begin{itemize}
  \item $\OPTL(I\arr{\ell'-3}{\ell'}) = 1$. \\ Indeed, $X = \{n-1, n-4\}$ and $Y = \{n-2,n-3\}$ are disjoint subsets of $I\arr{\ell'-3}{\ell'}$ and both sum to $2n-5$. Since $X$ contains $I[\ell']$, $(X,Y)$ is a feasible solution for $\OPTL(I\arr{\ell'-3}{\ell'})$, we obtain $\OPTL(I\arr{\ell'-3}{\ell'}) \le R(X,Y) = 1$.
\end{itemize}
  This demonstrates that equation (\ref{eq:reduction}) holds in this example. In fact, keeping the largest 4 items would have sufficed in this example. Below we prove that in general keeping the largest $\polylog 1/\eps$ items suffices.
  This example illustrates that the proof of equation (\ref{eq:reduction}) cannot start from (\ref{eq:optl}) and consider each subproblem $I\arr{1}{\ell}$ in isolation; instead in some cases we must change our focus to a different subproblem $I\arr{1}{\ell'}$.
  }
\end{example}


\medskip
The remainder of this section is devoted to the proof of Lemma~\ref{lem:reduction}.
We start by 
applying the pigeonhole principle to integral instances of \SSR.

\begin{obs}[Pigeonhole] \label{obs:pigeon}
  For any (multi-)set $I \subseteq \NN$ with $2^{|I|} > \Sum(I) + 1$ we have $\OPT(I) = 1$.
\end{obs}
\begin{proof}
  If $I$ is a multiset and any item appears at least twice in $I$, then by setting $X$ and $Y$ to contain only this item we obtain ratio $R(X,Y) = 1$, which implies $\OPT(I) = 1$. 
  So assume that $I$ is a set.
  Since each of the $2^{|I|}$ subsets of $I$ sums to a number in $\{0,\ldots,\Sum(I)\}$, by the pigeonhole principle there exist two different subsets $X',Y' \subseteq I$ with $\Sum(X') = \Sum(Y')$. We remove their intersection by considering the sets $X := X' \setminus (X' \cap Y')$ and $Y :=  Y' \setminus (X' \cap Y')$. Note that $X,Y$ are disjoint and satisfy $\Sum(X) = \Sum(X') - \Sum(X' \cap Y') = \Sum(Y') - \Sum(X' \cap Y') = \Sum(Y)$, proving that $R(X,Y) = 1$ and thus $\OPT(I) = 1$.
\end{proof}

We are now ready to prove Lemma~\ref{lem:reduction}.

\begin{proof}[Proof of Lemma~\ref{lem:reduction}]
  Recall that for any $x \ge 1$ we define the number $\LL(x)$ to be the smallest integer with $2^{\LL(x)} > \LL(x)\cdot x + 1$. Our goal is to show 
  \begin{align} \label{eq:reductiontoshow} 
    \OPT(I) \le \min_{1 \le \ell \le n} \OPTL(I\arr{\ell-\LL^2(4/\eps)}{\ell}) \le (1+\eps) \OPT(I).
  \end{align}
  Since for any $1 \le \ell \le n$ the set $I\arr{\ell - \LL^2(4/\eps)}{\ell}$ is a subset\footnote{Recall that we defined the notation $I\arr{i}{j} := I\arr{\max\{1,i\}}{\min\{n,j\}}$, so $I\arr{i}{j} \subseteq I$ always holds.} of $I$, by definition of $\OPT$ and $\OPTL$ we clearly have
  \[ \OPT(I) \le \OPT(I\arr{\ell - \LL^2(4/\eps)}{\ell}) \le \OPTL(I\arr{\ell - \LL^2(4/\eps)}{\ell}). \]
  This proves the inequality $\OPT(I) \le \min_{1 \le \ell \le n} \OPTL(I\arr{\ell - \LL^2(4/\eps)}{\ell})$. It remains to prove 
  \begin{align} \label{eq:toprove}
    \min_{1 \le \ell \le n} \OPTL(I\arr{\ell - \LL^2(4/\eps)}{\ell}) \le (1+\eps) \OPT(I).
  \end{align}
  We split the proof of this inequality into two cases. The first case will consider the situation that $I$ is locally dense, from which we will infer that $\OPTL(I\arr{\ell - \LL^2(4/\eps)}{\ell}) \le 1+\eps$ for \emph{some} $\ell$. Since $\OPT(I) \ge 1$ this yields~(\ref{eq:toprove}). In the second case we can then assume that $I$ is nowhere dense, and we will use this to show that $\OPTL(I\arr{\ell - \LL^2(4/\eps)}{\ell}) \le (1+\eps) \OPT(I\arr{1}{\ell})$ for \emph{all} $\ell$. Combining this with equation~(\ref{eq:optl}) then yields~(\ref{eq:toprove}). More details follow.
  
  \emph{Case 1: $I$ is locally dense, that is, there exists $z \in \RRp$ such that $|I \cap [z/2,z]| \ge \LL(4/\eps)$.} 
  We let $J \subseteq I \cap [z/2,z]$ consist of the $\LL(4/\eps)$ largest numbers in $I \cap [z/2,z]$. 
  Then we round up the items in $J$ to multiples of $\beta := \eps z / 4$, obtaining a multiset $\widetilde J$. Then the scaled set $\widetilde J / \beta$ consists of integers in the range $[1,4/\eps]$. The sum of all elements of $\widetilde J/\beta$ is thus bounded by $\Sigma(\widetilde J/\beta) \le |J| \cdot 4/\eps = \LL(4/\eps) \cdot 4/\eps$. By definition of $\LL(4/\eps)$ we thus have 
  $$ 2^{|\widetilde J/\beta|} = 2^{|J|} = 2^{\LL(4/\eps)} > \LL(4/\eps) \cdot 4/\eps + 1 \ge \Sigma(\widetilde J/\beta)+1. $$ 
  Hence, the Pigeonhole Observation~\ref{obs:pigeon} is applicable to the set $\widetilde J/\beta$ and yields $\OPT(\widetilde J/\beta) = 1$. This implies $\OPT(\widetilde J) = 1$, in other words, there exist disjoint $\widetilde X, \widetilde Y \subseteq \widetilde J$ with ratio $R(\widetilde X, \widetilde Y) = 1$. 
  
  Let $X,Y$ be the subsets of $J$ corresponding to $\widetilde X, \widetilde Y \subseteq \widetilde J$ before we rounded $J$ to $\widetilde J$. Since we rounded up to multiples of $\eps z/4$ and each element of $X,Y$ lies in $[z/2,z]$, we have $\Sigma(X) \le \Sigma(\widetilde X) \le (1 + \eps) \cdot \Sigma(X)$ and $\Sigma(Y) \le \Sigma(\widetilde Y) \le (1 + \eps) \cdot \Sigma(Y)$. For their ratio we thus obtain
  $$ 1 = \frac{\Sigma(\widetilde X)}{\Sigma(\widetilde Y)} \ge \frac{\Sigma(X)}{(1 + \eps) \cdot \Sigma(Y)}, $$
  which rearranges to $\Sigma(X) / \Sigma(Y) \le 1+\eps$. We can symmetrically show $\Sigma(Y) / \Sigma(X) \le 1+\eps$, and thus $R(X,Y) \le 1+\eps$.
  
  Let $\ell$ be such that $I[\ell] = \max(X \cup Y)$. Then we have $X,Y \subseteq I\arr{\ell - \LL(4/\eps)}{\ell}$ (since $X,Y \subseteq J$ and $J$ consists of $\LL(4/\eps)$ consecutive numbers in $I$). Thus, $(X,Y)$ is a feasible solution for the problem $\OPTL(I\arr{\ell - \LL(4/\eps)}{\ell})$, so we obtain
  $$\OPTL(I\arr{\ell - \LL(4/\eps)}{\ell}) \le R(X,Y) \le 1+\eps \le (1+\eps) \OPT(I), $$ 
  where the last inequality used $\OPT(I) \ge 1$. 
  Since $\LL(4/\eps) \le \LL^2(4/\eps)$, we have proved inequality~(\ref{eq:toprove}):
  \begin{align*}
    \min_{1 \le \ell' \le n} \OPTL(I\arr{\ell' - \LL^2(4/\eps)}{\ell'}) 
    &\le \OPTL(I\arr{\ell - \LL^2(4/\eps)}{\ell}) \\ 
    &\le \OPTL(I\arr{\ell - \LL(4/\eps)}{\ell}) \le (1+\eps) \OPT(I).
  \end{align*}
  
  
  \emph{Case 2: $I$ is nowhere dense, that is, for all $z \in \RRp$ we have $|I \cap [z/2,z]| < \LL(4/\eps)$.} 
  In this case we show that $\OPTL(I\arr{\ell - \LL^2(4/\eps)}{\ell}) \le (1+\eps) \OPTL(I\arr{1}{\ell})$ holds for \emph{all} $\ell$. This is clear for $\ell \le \LL^2(4/\eps)$, as then $I\arr{\ell - \LL^2(4/\eps)}{\ell} = I\arr{1}{\ell}$. So let $\LL^2(4/\eps) < \ell \le n$. 
  By the case assumption (i.e. $|I \cap [z/2,z]| < \LL(4/\eps)$ for all $z \in \RRp$) it follows that $I[\ell - 1 - k \cdot \LL(4/\eps)] \le 2^{-k} I[\ell - 1]$ for any $k \in \NN$. In particular, we have 
  $$ I[\ell - 1 - \LL^2(4/\eps)] \le 2^{-\LL(4/\eps)} I[\ell-1] < \frac \eps{4 \LL(4/\eps)} I[\ell-1], $$ 
  by definition of $\LL(4/\eps)$.
  Moreover, for any $z \in \RRp$ we have $\Sum(I \cap (z/2,z]) \le z \cdot |I \cap (z/2,z]| < z \LL(4/\eps)$. By a geometric series we thus obtain 
  \[ \Sum(I \cap (0,z]) = \sum_{i \ge 0} \Sum(I \cap (z/2^{i+1},z/2^i]) < \sum_{i \ge 0} 2^{-i} z \LL(4/\eps) = 2 z \LL(4/\eps). \]
  Combining these facts yields 
  \begin{align*} 
    \Sum(I\arr{1}{\ell - 1 - \LL^2(4/\eps)}) &= \Sum(I \cap (0, I[\ell - 1 - \LL^2(4/\eps)]]) \\ 
    &\le 2 \LL(4/\eps) \cdot I[\ell - 1 - \LL^2(4/\eps)] \le \tfrac \eps 2 \cdot I[\ell-1]. 
  \end{align*}
  Now consider an optimal solution for $\OPTL(I\arr{1}{\ell})$, that is, disjoint $X,Y \subseteq I\arr{1}{\ell}$ with $\max(X \cup Y) = I[\ell]$ attaining $R(X,Y) = \OPTL(I\arr{1}{\ell})$. Let $X' := X \cap I\arr{\ell - \LL^2(4/\eps)}{\ell}$ and $Y' := Y \cap I\arr{\ell - \LL^2(4/\eps)}{\ell}$. Note that $\Sum(X) - \Sum(X') \le \Sum(I\arr{1}{\ell - 1 - \LL^2(4/\eps)}) \le \tfrac \eps 2 \cdot I[\ell-1]$, and thus
  \[ \Sum(X') \le \Sum(X) \le \Sum(X') + \tfrac \eps 2 \cdot I[\ell-1] \qquad \text{and similarly} \qquad \Sum(Y') \le \Sum(Y) \le \Sum(Y') + \tfrac \eps 2 \cdot I[\ell-1]. \]
  We can thus bound the ratio $R(X',Y')$ by
  \begin{align} \label{eq:tocontinue}
    R(X',Y') = \max\left\{ \frac{\Sum(X')}{\Sum(Y')}, \frac{\Sum(Y')}{\Sum(X')} \right\} 
    \le \max\left\{ \frac{\Sum(X)}{\Sum(Y) - \tfrac \eps 2 \cdot I[\ell-1]}, \frac{\Sum(Y)}{\Sum(X) - \tfrac \eps 2 \cdot I[\ell-1]} \right\}. 
  \end{align}
  By symmetry we can assume that $\Sum(X) \ge \Sum(Y)$, so $R(X,Y) = \Sum(X)/\Sum(Y)$. Since $\max(X \cup Y) = I[\ell]$ we thus have $\Sum(X) \ge I[\ell]$. 
  Since $X'' = \{I[\ell]\}, Y'' = \{I[\ell-1]\}$ form a feasible solution with ratio $I[\ell]/I[\ell-1]$, the optimal solution $X,Y$ satisfies $R(X,Y) = \Sum(X)/\Sum(Y) \le I[\ell]/I[\ell-1]$. Rearranging this yields $\Sum(Y) \ge \Sum(X) \cdot I[\ell-1] / I[\ell] \ge I[\ell-1]$.
  Hence, $\Sum(X) \ge \Sum(Y) \ge I[\ell-1]$.
  Using this bound in (\ref{eq:tocontinue}) yields
  \[ R(X',Y') \le \frac{\Sum(X)}{\Sum(Y) \cdot(1 - \tfrac \eps 2)} = \frac{R(X,Y)}{1 - \tfrac \eps 2} \le (1+\eps) R(X,Y), \]
  where we used $1/(1- \tfrac \eps 2) \le 1+\eps$ for any $\eps \in (0,1)$.
  Since $(X',Y')$ is a feasible solution for the problem $\OPTL(I\arr{\ell - \LL^2(4/\eps)}{\ell})$, we obtain 
  $$ \OPTL(I\arr{\ell - \LL^2(4/\eps)}{\ell}) \le R(X',Y') \le (1+\eps) R(X,Y) = (1+\eps) \OPTL(I\arr{1}{\ell}). $$
  Since this holds for all $1 \le \ell \le n$, we have
  $$ \min_{1 \le \ell \le n} \OPTL(I\arr{\ell - \LL^2(4/\eps)}{\ell})
  \le (1+\eps) \min_{1 \le \ell \le n} \OPTL(I\arr{1}{\ell}) \stackrel{(\ref{eq:optl})}{=} (1+\eps) \OPT(I). $$
  Thus, also in this case we established inequality~(\ref{eq:toprove}), which finishes the proof.
\end{proof}

\section{Solving \SSRL with Sublinear Dependence on \boldmath$1/\eps$} \label{sec:improved}

For the whole section we are given an \SSRL instance $I \subset \RRp$ of size $n$ and a number $\eps \in (0,1)$. In light of Lemma~\ref{lem:reduction} reducing the number of items to $\polylog 1/\eps$, we do not care about the running time dependence on $n$; any polynomial dependence is fine. However, our goal is to decrease the running time dependence on $1/\eps$ to sublinear, more precisely to $1/\eps^{0.93856}$. In other words, in this section we want to design an \SSRL algorithm running in time $\Os(1/\eps^{0.93856})$, where the $\Os$ notation hides polynomial factors in $n$.

\subsection{Algorithm Overview}
\label{sec:overview}

Our algorithm is an intricate combination of a pigeonhole argument with ideas from rounding, sumsets, and computational geometry. 

\paragraph{A Failed Approach}
Let us first discuss a simple approach and why it fails. 
We want to use two ideas. Idea (1): It is safe to round all items to multiples of $\max(I) \cdot \eps/n$. Indeed, we want to solve an \SSRL instance and thus are only interested in subsets containing the largest item $\max(I)$, and rounding to multiples of $\max(I) \cdot \eps/n$ changes the ratio $R(X,Y)$ of sets $X,Y$ with $\max(I) \in X \cup Y$ by at most a factor $1+\eps$. 
Idea (2): The pigeonhole principle guarantees the existence of two equal subset sums if we have more than a logarithmic number of items. In particular, it turns out that if after rounding we have more than $\log(\Theta(n^2/\eps))$ items, then two disjoint subsets have equal rounded sum. 
Thus, we are done unless $n \le \log(\Theta(n^2/\eps))$, which solves to $n \le \log 1/\eps + \Theta(\log \log 1/\eps)$. 
But then $n$ is small enough to run an exact algorithm in time $O(1.7088^n) = O(1/\eps^{0.773})$, which is sublinear in $1/\eps$ as desired. 

The issue is that ideas (1) and (2) do not work together nicely. Indeed, the equal subset sums that are guaranteed by the pigeonhole principle not necessarily contain $\max(I)$, and thus can have much smaller sum than $\max(I)$. But if they have much smaller sum, then the rounding with respect to multiples of $\max(I) \cdot \eps/n$ was too coarse to say anything about their true sums even if their rounded sums are equal.\footnote{It is easy to construct instances where this situation appears. Consider for example $\eps = 1/100$ and $I = \{100^{0}, 100^{1}, \ldots, 100^{8}, 100^{9}, 100^{100} \}$, where after rounding to multiples of $\max(I) \cdot \eps/n$ all items (except $\max(I)$) are rounded to the same value, and the pigeonhole principle yields a trivial solution on the rounded instance, but the unrounded instance has no solution with ratio $<50$.}

Therefore, this simple approach fails. Our actual approach is significantly more complicated and subtle, and it makes the combination of ideas (1) and (2) work despite the issues with the simple approach, using a variety of additional tricks. It can be instructive to consider Lemma~\ref{lem:refinedpigeon}, which guarantees that some ``large'' item (i.e., some item in $T$) is selected by the pigeonhole solution, and thus can be seen as a remedy for the issue that the pigeonhole principle does not guarantee to select the maximum item.

\paragraph{Our Actual Approach}

We start by splitting $I$ into its largest $\tau$ items $T$ (top) and the remaining items $B = I\setminus T$ (bottom). 
In one part of our algorithm, we would like to perform meet in the middle on $T$ and~$B$. Meet in the middle is standard for the Equal Subset Sum problem: Compute the sets $P' := \{\Sum(X) - \Sum(Y) \mid \text{disjoint } X,Y \subseteq T\}$ and $Q' := \{\Sum(X) - \Sum(Y) \mid \text{disjoint } X,Y \subseteq B\}$ in sorted order. Then for each $p \in P'$ perform binary search over $Q'$ to check whether $p \in Q'$. If we find a match, then we obtain sets $X_T,Y_T \subseteq T$ and $X_B, Y_B \subseteq B$ with $\Sum(X_T) - \Sum(Y_T) = \Sum(Y_B) - \Sum(X_B)$. The sets $X = X_T \cup X_B$ and $Y = Y_T \cup Y_B$ thus form an Equal Subset Sum solution. 
However, for Subset Sum Ratio there is an additional difficulty: Sets with closer sums $\Sum(X_T) - \Sum(Y_T)$ and $\Sum(Y_B) - \Sum(X_B)$ could have smaller sums $\Sum(X_T) + \Sum(X_B)$ and $\Sum(Y_T) + \Sum(Y_B)$, and thus could have a worse ratio. Therefore, the goal is not simply to find two numbers in $P'$ and $Q'$ that are as close as possible, as in general there can be a tradeoff between closeness and magnitude. 
It turns out that to generalize meet in the middle we have to consider the point sets $P := \{(\Sum(X) - \Sum(Y), \Sum(X) + \Sum(Y)) \mid \text{disjoint } X,Y \subseteq T\}$ and $Q := \{(\Sum(X) - \Sum(Y), \Sum(X) + \Sum(Y)) \mid \text{disjoint } X,Y \subseteq B\}$. (Moreover, if two points in $P$ have the same $x$-coordinate then we can remove the lower of the two points from~$P$, similarly for $Q$.) The goal then turns out to be maximizing $\frac{p_y + q_y}{|p_x+q_x|}$ over all points $(p,q) \in P \times Q$. By further rearrangements, we reduce this to a computational geometry problem on slopes and convex hulls that can be solved in near-linear time. 


The resulting algorithm solves \SSRL in time $\Os(|P| + |Q|)$. For $P$ we use the simple bound $|P| \le 3^\tau$, since each of the $\tau$ elements of $T$ can be in $X$ or in $Y$ or in neither of them. 
To be able to bound $|Q|$, we need to combine this algorithm with rounding. It turns out that we can afford to round items to multiples of $\frac \eps n \max(T)$ while still computing a $(1+\eps)$-approximation.
Since all coordinates of points in $Q$ are bounded by $O(n \cdot \max(B)) \le O(n \cdot \min(T))$, after this rounding we obtain a bound of $|Q| = \Os(1/\eps \cdot \min(T)/\max(T))$, which is sublinear in $1/\eps$ whenever $\min(T) \ll \max(T)$. This provides one situation in which we have a sufficiently good bound on $|Q|$.
Another such situation happens if (after rounding) $B$ creates few different subset sums. We can bound $|Q|$ by the square of the number of subset sums of $B$, and thus we obtain a fast algorithm if the number of different subset sums is small.

The remaining case is that $\min(T)$ is not much smaller than $\max(T)$ and (the rounded) $B$ generates many different subset sums. In this case, we make use of a refined pigeonhole argument: We show that if the product of the number of subset sums of (the rounded) $B$ and the number of subset sums of (the rounded) $T$ is large enough, then they are guaranteed to have a (rounded) solution $\tX,\tY$ with ratio $R(\tX,\tY)=1$ containing at least one item of $T$. Before rounding, these sets $\tX,\tY$ then correspond to sets $X,Y$ with $R(X,Y) \le 1+\eps$. 
We show that this pigeonhole argument is always applicable in the remaining case that we need to consider. 
Since ratio $R(\tX,\tY)=1$ means $\Sum(\tX) = \Sum(\tY)$, in this part of the algorithm we can use standard meet in the middle to actually find such sets $\tX,\tY$.

A major technical difficulty is that the pigeonhole argument does not guarantee a solution containing $\max(I)$ --- it only guarantees picking some item in $T$. This in particular means that in the pigeonhole-based algorithm we can only afford to round to multiples of $\frac \eps n \min(T)$ (instead of $\frac \eps n \max(T)$ as in the geometric algorithm). That the two parts of the algorithm use different roundings causes significant difficulties.

Finally, our analysis of the above algorithms only works when the optimal value is at most a constant. We handle the case of large optimal value by presenting an algorithm for \SSRL that computes a solution of ratio at most $\max\{\sqrt{2}, \OPTL(I)\}$. This yields a simple $\sqrt{2}$-approximation algorithm, but more importantly it solves instances with optimal value at least $\sqrt{2}$ to optimality.

\paragraph{Organization}
After preparations in Section~\ref{sec:prepsumsets}, we present our refined pigeonhole argument and its application to \SSRL in Section~\ref{sec:refinedpigeon}. Then we present our generalized meet in the middle algorithm using computational geometry in Section~\ref{sec:compgeom}.  After preparations on rounding in Section~\ref{sec:preprounding} we then combine the two parts in Section~\ref{sec:combination}. Finally, we handle the case of large optimal value in Section~\ref{sec:removerequirement}.

\subsection{Preparations on Sumsets}
\label{sec:prepsumsets}

Throughout this section we use $\Os$ notation to suppress any $\textup{poly}(n)$ factors (as in our final running time they will only contribute $\polylog(1/\eps)$ factors).

We denote by $\cS(Z)$ the set (not multi-set!) of all subset sums of $Z$, that is,
\[ \cS(Z) := \{ \Sum(Y) \mid Y \subseteq Z \}. \]
We will make use of the standard fact that $\cS(Z)$ can be computed in output-sensitive time:
\begin{lemma} \label{lem:computecS}
  Given $Z \subset \RRp$ of size $n$, we can compute $\cS(Z)$ in time $\Os(|\cS(Z)|)$. Moreover, given $Z \subset \RRp$ of size $n$ and an integer $S$, we can decide whether $|\cS(Z)| \ge S$ in time $\Os(S)$.
\end{lemma}
\begin{proof}
  Write $Z = \{z_1,\ldots,z_n\}$ and $Z_i := \{z_1,\ldots,z_i\}$. Note that we can compute $\cS(Z_1) = \{0,z_1\}$ in constant time, and we can compute $\cS(Z_i)$ from $\cS(Z_{i-1})$ as
  \begin{align} \label{eq:cSfromprevious}
    \cS(Z_i) = \cS(Z_{i-1}) \cup \{s + z_i \mid s \in \cS(Z_{i-1})\}. 
  \end{align}
  If we know $\cS(Z_{i-1})$ in sorted order, then we know both sets on the right hand side in sorted order, so by one merge operation we can compute $\cS(Z_i)$ in time $O(|\cS(Z_i)|)$. Hence, we can compute $\cS(Z)$ in time $O(\sum_{i=1}^n |\cS(Z_i)|) = O(n \cdot |\cS(Z)|) = \Os(|\cS(Z)|)$.
  
  To decide $|\cS(Z)| \ge S$ we run the above algorithm, but we abort once we computed a set $\cS(Z_i)$ of size at least $S$. Since equation (\ref{eq:cSfromprevious}) implies that $|\cS(Z_i)| \le 2 |\cS(Z_{i-1})|$, we abort with a set of size $|\cS(Z_i)| \le 2S$. The running time is thus $O(n |\cS(Z_i)|) = O(n S) = \Os(S)$.
\end{proof}

\begin{definition} \label{def:pointsets}
  For a point in the plane $p \in \RR^2$ we write $p_x$ and $p_y$ for its $x$- and $y$-coordinates.
  We say that point $q \in \RR^2$ \emph{dominates} point $p \in \RR^2$ if $q_x = p_x$ and $q_y > p_y$. 
  For a set of points in the plane $P \subset \RR^2$ we denote the set of all non-dominated points in $P$ by $\opU(P)$, that is,
\[ \opU(P) := \{ p \in P \mid \not\exists q \in P \colon q_x=p_x \text{ and } q_y > p_y \}. \]
  For a set $I \subset \RRp$, we define the point sets
  \begin{align*} 
    P(I) &:= \opU( \{(\Sum(X) - \Sum(Y), \Sum(X) + \Sum(Y)) \mid \text{disjoint } X,Y \subseteq I \} ), \\
    P_L(I) &:= \opU( \{(\Sum(X) - \Sum(Y), \Sum(X) + \Sum(Y)) \mid \text{disjoint } X,Y \subseteq I, \, \max(I) \in X \cup Y \} ), \\
    P_N(I) &:= \opU( \{(\Sum(X) - \Sum(Y), \Sum(X) + \Sum(Y)) \mid \text{disjoint } X,Y \subseteq I, \, X \cup Y \ne \emptyset \} ).
  \end{align*}
\end{definition}

\begin{lemma}[Point Set Generation] \label{lem:sumgen}
  Given a set $I \subset \RRp$ of size $n$, we can compute the point sets $P(I), P_L(I)$, and $P_N(I)$ in sorted order (sorted by $x$-coordinate) in time $\Os(|P(I)|)$.
\end{lemma}
\begin{proof}
  We generate $P(I)$ from $P' := P(I\arr{1}{n-1})$ as
  \[ P(I) = \opU\big(  P' \cup \{(p_x + I[n], p_y + I[n]) \mid p \in P'\} \cup \{(p_x - I[n], p_y + I[n]) \mid p \in P'\} \big). \]
  If all three sets on the right hand side are sorted by $x$-coordinate, we can generate $P(I)$ in sorted order by merging three sorted sets and removing dominated points. Thus, the running time for this step is $O(|P'|)$, and the total running time to compute $P(I)$ is $O\left(\sum_{i=1}^{n-1} |P(I\arr{1}{i})| \right) \le O(n \cdot |P(I)|)$. 
  
  We similarly generate the set $P_L(I)$ from $P' = P(I\arr{1}{n-1})$ as
  \[ P_L(I) = \opU\big( \{(p_x + I[n], p_y + I[n]) \mid p \in P'\} \cup \{(p_x - I[n], p_y + I[n]) \mid p \in P'\} \big). \]
  
  We similarly generate the set $P_N(I)$ from $P' = P(I\arr{1}{n-1})$ as $P_N(\emptyset) = \emptyset$ and for $n \ge 1$:
  \[ P_N(I) = \opU\big( P_N(I\arr{1}{n-1}) \cup \{(p_x + I[n], p_y + I[n]) \mid p \in P'\} \cup \{(p_x - I[n], p_y + I[n]) \mid p \in P'\} \big). \]
  
  Clearly the same running time bounds apply to $P_L(I)$ and $P_N(I)$.
\end{proof}

\subsection{A Refined Pigeonhole Principle and its Application to \SSRL}
\label{sec:refinedpigeon}

Our improved algorithm starts with the following refinement of the pigeonhole principle (which should be compared with the basic Observation~\ref{obs:pigeon}). 

\begin{lemma}[Refined Pigeonhole] \label{lem:refinedpigeon}
  Let $I \subseteq \NN$ be partitioned into $I = B \cup T$. If $|\cS(B)| \cdot |\cS(T)| > \Sum(B) + \Sum(T) + 1$, then there exist disjoint sets $X,Y \subseteq I$ with $\Sum(X) = \Sum(Y)$ and $(X \cup Y) \cap T \ne \emptyset$.
\end{lemma}
\begin{proof}
  Consider all pairs $(x_B,x_T) \in \cS(B) \times \cS(T)$. Note that for each such pair we have $x_B+x_T \in \{0,\ldots,\Sum(B)+\Sum(T)\}$. By the assumption $|\cS(B)| \cdot |\cS(T)| > \Sum(B) + \Sum(T) + 1$ and the pigeonhole principle, there exist two different pairs $(x_B,x_T),(y_B,y_T) \in \cS(B) \times \cS(T)$ with equal sum $x_B+x_T = y_B+y_T$. These pairs correspond to sets $X_B,Y_B \subseteq B, X_T,Y_B \subseteq T$ with $\Sum(X_B) = x_B$, $\Sum(X_T) = x_T$, $\Sum(Y_B) = y_B$, $\Sum(Y_T) = y_T$. 
  We let $X' := X_B \cup X_T$ and $Y' := Y_B \cup Y_T$. 
  Removing their intersection, we obtain disjoint sets $X := X' \setminus (X' \cap Y')$ and $Y := Y' \setminus (X' \cap Y')$. These disjoint sets have equal sum, since
  \begin{align*}
    \Sum(X) = \Sum(X') - \Sum(X' \cap Y') &= x_B + x_T - \Sum(X' \cap Y') \\
    &= y_B + y_T - \Sum(X' \cap Y') = \Sum(Y') - \Sum(X' \cap Y') = \Sum(Y). 
  \end{align*}
  Finally, we argue that $(X \cup Y) \cap T \ne \emptyset$. Since $x_B + x_T = y_B + y_T$ and the pairs $(x_B,x_T), (y_B,y_T)$ are different, it follows that $x_T \ne y_T$. We thus have
  \[ \Sum(X \cap T) = x_T - \Sum(X' \cap Y' \cap T) \ne y_T - \Sum(X' \cap Y' \cap T) = \Sum(Y \cap T). \]
  It follows that at least one of $X \cap T$ or $Y \cap T$ must be nonempty, and thus $(X \cup Y) \cap T \ne \emptyset$.
\end{proof}

We now use the refined pigeonhole principle to solve a very specific special case of the \SSRL problem.

\begin{lemma} \label{lem:mainpigeon}
  Let $I \subset \RRp$ be a (multi-)set of size $n$ consisting of numbers that are integer multiples of $\beta \in \RRp$, and let $\tau \in \mathbb{N}$. Let $T$ be the largest $\tau$ items of $I$, let $B := I \setminus T$ be the remaining items.
  If $|\cS(B)| > n \max(I) / (\beta 2^\tau) + 1$ then there exist disjoint $X,Y \subseteq I$ with $R(X,Y) = 1$ and $(X \cup Y) \cap T \ne \emptyset$, and we can compute such $X,Y$ in time $\Os(\max(I)^2 / (\beta^2 4^\tau) + 3^\tau)$.
\end{lemma}

\begin{proof}
Note that scaling all numbers by $\beta$ reduces to the case $\beta = 1$, so it suffices to prove the lemma in the case $\beta=1$, which is equivalent to $I \subset \mathbb{N}$.

The algorithms works as follows:
\begin{enumerate}
\item Enumerate all pairs $(X,Y)$ of disjoint subsets of $T$. This can be done naively in time $\Os(3^\tau)$, as each element of $T$ can go into $X$ or into $Y$ or into neither. If we find non-empty disjoint subsets $X,Y \subseteq T$ with $\Sum(X) = \Sum(Y)$, then we have $R(X,Y) = 1$ and thus we are done. 

Hence, from now on we can assume that no non-empty disjoint subsets $X,Y \subseteq T$ satisfy $\Sum(X) = \Sum(Y)$. 
This implies that there are no two different subsets $X',Y' \subseteq T$ with equal sum $\Sum(X') = \Sum(Y')$. Indeed, for any such sets $X',Y'$ we could remove their intersection, obtaining disjoint sets $X := X' \setminus (X' \cap Y')$ and $Y := Y' \setminus (X' \cap Y')$ with the same sum $\Sum(X) = \Sum(X') - \Sum(X' \cap Y') = \Sum(Y') - \Sum(X' \cap Y') = \Sum(Y)$. Hence, from now on no two subset sums of $T$ are equal, that is, we have $|\cS(T)| = 2^\tau$. 

\item Write $B = \{b_1,\ldots,b_m\}$ and $B_j := \{b_1,\ldots,b_j\}$. For $j = 1,2,\ldots$ compute the set $\cS(B_j)$ by Lemma~\ref{lem:computecS}, and stop at the first $j$ with $|\cS(B_j)| > n \max(I) / 2^\tau + 1$. Note that by the assumption in the lemma statement, we indeed stop for some $j \le m$. Since $|\cS(B_j)| \le 2 |\cS(B_{j-1})|$ (which follows from equation~(\ref{eq:cSfromprevious})), we also have $|\cS(B_j)| = O(n \max(I) / 2^\tau + 1)$.

\item Compute $P := \{\Sum(X) - \Sum(Y) \mid \text{disjoint } X,Y \subseteq B_j \}$ and $Q := \{\Sum(X) - \Sum(Y) \mid \text{disjoint } X,Y \subseteq T, X \cup Y \ne \emptyset \}$. These sets can be computed in time $\Os(|P|)$ and $\Os(|Q|)$, by applying Lemma~\ref{lem:sumgen} to compute $P(B_j)$ and $P_N(T)$ and then removing the $y$-coordinates of the computed points.

\item 
From $|\cS(B_j)| > n \max(I) / 2^\tau + 1$ and $|\cS(T)| = 2^\tau$ we obtain $|\cS(B_j)| \cdot |\cS(T)| > n \max(I) + 1 \ge \Sum(B_j) + \Sum(T) + 1$. That is, the Refined Pigeonhole Lemma~\ref{lem:refinedpigeon} is applicable, and it shows existence of disjoint sets $X,Y \subseteq B_j \cup T$ with $\Sum(X) = \Sum(Y)$ and $(X \cup Y) \cap T \ne \emptyset$. Consider $p := \Sum(X \cap B_j) - \Sum(Y \cap B_j) \in P$ and $q := \Sum(Y \cap T) - \Sum(X \cap T) \in Q$. Then $p - q = \Sum(X) - \Sum(Y) = 0$, so it follows that $P \cap Q \ne \emptyset$.

We use meet in the middle to find some $x \in P \cap Q$. In other words, we iterate over all $p \in P$ and for each $p$ we perform binary search over $Q$ to test whether $p \in Q$. 
After finding $x \in P \cap Q$, by appropriate bookkeeping during the construction of $P$ and $Q$ we can reconstruct disjoint sets $X_B, Y_B \subseteq B_j$ with $x = \Sum(X_B) - \Sum(Y_B)$, as well as disjoint sets $X_T,Y_T \subseteq T$ with $x = \Sum(Y_T) - \Sum(X_T)$ and $X_T \cup Y_T \ne \emptyset$. Setting $X := X_B \cup X_T$ and $Y := Y_B \cup Y_T$ now yields disjoint sets $X,Y \subseteq I$ with $R(X,Y) = 1$ and $(X \cup Y) \cap T \ne \emptyset$. This finishes the algorithm.
\end{enumerate}
We argued correctness in the algorithm description. Regarding running time, note that step 1 runs in time $\Os(3^\tau)$, step 2 runs in time $\Os(|\cS(B_j)|) = \Os(\max(I)/2^\tau + 1)$, and steps 3 and 4 run in time $\Os(|P| + |Q|)$. We can easily bound $|Q| \le 3^\tau$. For $|P|$ we observe that we can write each element of $P$ as a difference of two elements of $\cS(B_j)$, and thus $|P| \le |\cS(B_j)|^2 = O(n^2 \max(I)^2 / 4^\tau + 1)$. 
In total, we thus obtain running time $\Os(\max(I)^2 / 4^\tau + 3^\tau)$.
\end{proof}

\subsection{Generalized Meet in the Middle via Computational Geometry}
\label{sec:compgeom}

We will use some tools from computational geometry, specifically convex hulls and tangent lines of convex hulls. These tools can be used to determine the maximum slope of any line through two points in two given point sets, as shown by the following lemma.

\begin{lemma}[\cite{MehlhornR10}] \label{lem:mehlhorn}
  Given sets $P,Q \subset \RR^2$ of $n$ points in the plane, in time $O(n \log n)$ we can compute for each point $p \in P$ the maximum slope of a line through $p$ and a point $q \in Q$ to the right of $p$, that is, we compute
  \[ s_p := \max\left\{ \frac{q_y - p_y}{q_x - p_x} \;\bigg|\; \,q\in Q, \, p_x < q_x \right\} \text{ for each } p \in P. \]
\end{lemma}
\begin{proof}
  Mehlhorn and Ray~\cite{MehlhornR10} describe a simple $O(n \log n)$-time algorithm for this problem (see their Section~3.1). They also describe a faster $O(n)$-time algorithm assuming that the points are given in sorted order, but this log-factor improvement is irrelevant in our context.
  For convenience, in the following we sketch the simple $O(n \log n)$-time algorithm by Mehlhorn and Ray. 
  
  We first sort $P \cup Q$ by $x$-coordinates.
  Then we scan the points $P \cup Q$ from right to left, maintaining an array containing the upper convex hull $\cal C$ of the seen points in~$Q$. There are two types of events to process:
  
  (1) When we reach a point $p \in P$, then we want to find the point $q \in Q$ to the right of $p$ maximizing the slope of the line through $p$ and $q$. This point must lie on $\cal C$, more precisely we are looking for the tangent point on $\cal C$ of the tangent line touching $\cal C$ and passing through $p$. This tangent point can be found by standard binary search in time $O(\log n)$. 
  
  (2) When we reach a point $q \in Q$, then we update $\cal C$ by using the standard incremental algorithm for computing the convex hull of points in the plane: We iterate over the points in $\cal C$ starting from the leftmost, to enumerate all points that are removed from the upper hull after inserting $q$. Then we add $q$ to $\cal C$. 
  
  Note that every point is added exactly once to the convex hull, and thus every point is removed at most once from the convex hull. It follows that the total running time of the updates in (2) is $O(n)$.
  In total, the running time is $O(n \log n)$.
\end{proof}

\begin{lemma} \label{lem:pointsetsintermediate}
  Given sets $P,Q \subset \RR^2$ of $n$ points in the plane, in time $O(n \log n)$ we can compute the maximum of $(q_y - p_y)/|q_x - p_x|$ over all $(p,q) \in P \times Q$. 
\end{lemma}
\begin{proof}
  We sort the points by $x$-coordinate to determine whether any $p \in P, q \in Q$ have equal $x$-coordinate, and thus whether the result is $\infty$. 
  Thus, assume from now on that there are no such points, and the result is finite. We split the term to be maximized into:
  \[ \alpha := \max\left\{ \frac{q_y - p_y}{q_x - p_x} \;\bigg|\; p \in P, q \in Q, p_x < q_x\right\} \quad\text{and}\quad \beta := \max\left\{ \frac{q_y - p_y}{p_x - q_x} \;\bigg|\; p \in P, q \in Q, q_x < p_x\right\}. \]
  In order to compute $\alpha$, we apply Lemma~\ref{lem:mehlhorn} on $P$ and $Q$ 
  and take the maximum over all computed values~$s_p$, $p \in P$. To compute $\beta$, we apply Lemma~\ref{lem:mehlhorn} on $P' := \{(-p_x,p_y) \mid p \in P\}$ and $Q' := \{(-q_x,q_y) \mid q \in Q\}$ and take the maximum over all computed values $s_p$, $p \in P'$.
  Finally, we return the maximum of $\alpha$ and $\beta$.
\end{proof}

\begin{lemma} \label{lem:pointsets}
  Given sets $P,Q \subset \RR^2$ of $n$ points in the plane satisfying $r_y \ge r_x$ for all $r \in P \cup Q$, in time $O(n \log n)$ we can compute the minimum of $(p_y + q_y + |p_x + q_x|)/(p_y + q_y - |p_x + q_x|)$ over all $(p,q) \in P \times Q$. 
\end{lemma}
\begin{proof}
  We rearrange
  \[ \frac{p_y + q_y + |p_x + q_x|}{p_y + q_y - |p_x + q_x|} = 1 + 2\frac{|p_x + q_x|}{p_y + q_y - |p_x + q_x|} = 1 + 2 \left( \frac{p_y + q_y}{|p_x + q_x|} - 1 \right)^{-1}.  \]
  Denote by $\mu$ the maximum of $(p_y + q_y)/|p_x + q_x|$ over all $p \in P, q \in Q$. It follows that our goal is to compute $1 + 2/(\mu - 1)$. Note that $\mu$ can be computed by applying Lemma~\ref{lem:pointsetsintermediate} to $P' := \{(-p_x,-p_y) \mid p \in P\}$ and $Q$. 
\end{proof}

Now we use this geometric tool for the \SSRL problem, to obtain a generalization of the meet in the middle approach for Equal Subset Sum. Specifically, we apply the geometric tool to the point sets $P(B)$ and $P_L(T)$ (recall Definition~\ref{def:pointsets}).

\begin{lemma} \label{lem:maingeometric}
  Given a multiset $I \subset \RRp$ of size $n$ that is partitioned into $I = T \cup B$ with $\max(I) \in T$, write $S := |P(B)| + |P_L(T)|$. Then in time $\Os(S)$ we can compute $\OPTL(I)$, as well as an optimal solution realizing $\OPTL(I)$ (i.e., disjoint subsets $X,Y \subseteq I$ with $\max(I) \in X \cup Y$ and $R(X,Y) = \OPTL(I)$). 
\end{lemma}
\begin{proof}
  The algorithm works as follows.
  \begin{enumerate}
  \item Generate the point sets $P(B)$ and $P_L(T)$ using Lemma~\ref{lem:sumgen}.
  \item Run Lemma~\ref{lem:pointsets} on $P := P(B)$ and $Q := P_L(T)$ to compute the value 
  $$\mu = \min\left\{ \frac{p_y + q_y + |p_x + q_x|}{p_y + q_y - |p_x + q_x|} \;\bigg|\; p \in P, q \in Q\right\}.$$ 
  \item Return the number $\mu$.
  \end{enumerate}
  To show correctness of this algorithm, we need to show that $\mu = \OPTL(I)$. To this end, consider disjoint $X,Y \subseteq I$ with $\max(I) \in X \cup Y$. Consider the corresponding points 
  \begin{align*}
    p &= \big(\Sum(X \cap B) - \Sum(Y \cap B),\, \Sum(X \cap B) + \Sum(Y \cap B)\big), \\
    q &= \big(\Sum(X \cap T) - \Sum(Y \cap T),\, \Sum(X \cap T) + \Sum(Y \cap T)\big).
  \end{align*}
  Observe that $p \in P$ and $q \in Q$. Thus, we described a mapping $\Phi$ from pairs $(X,Y)$, where $X,Y$ are disjoint subsets of $I$ with $\max(I) \in X \cup Y$, to pairs $(p,q) \in P \times Q$. Observe that the mapping~$\Phi$ is surjective.
  
  We rewrite the definition of $R(X,Y)$ to obtain
  \[ R(X,Y) = \frac{\max\{\Sum(X),\Sum(Y)\}}{\min\{\Sum(X),\Sum(Y)\}} = \frac{\Sum(X) + \Sum(Y) + |\Sum(X) - \Sum(Y)|}{\Sum(X) + \Sum(Y) - |\Sum(X) - \Sum(Y)|} = \frac{p_y + q_y + |p_x + q_x|}{p_y + q_y - |p_x + q_x|}. \]
  This shows that the surjective mapping $\Phi$ maintains the value $R(X,Y)$. It follows that by minimizing $\frac{p_y + q_y + |p_x + q_x|}{p_y + q_y - |p_x + q_x|}$ over all $(p,q) \in P \times Q$ we compute the minimum value $R(X,Y)$ over all disjoint $X,Y \subseteq I$ with $\max(I) \in X \cup Y$. As the latter is the definition of $\OPTL(I)$, we obtain $\mu = \OPTL(I)$. This proves correctness of the algorithm.
  
  \smallskip
  Regarding running time, note that step 1 runs in time $\Os(S)$ and step 2 runs in time $O(S \log S)$. By the simple bounds $|P_L(J)| \le |P(J)| \le 3^{|J|}$ for any set $J$, we observe that $\log S \le O(\log(3^n)) = O(n)$ and thus step 2 runs in time $O(nS) = \Os(S)$ as well.
  
  \smallskip
  Note that the above algorithm only computes the value $\OPTL(I)$, but not an optimal solution realizing this value. It is easy to augment the geometric tool (Lemma~\ref{lem:pointsets}) to not only compute the value $\mu$, but also the points $(p,q)$ realizing this value, and to augment the point set generation (Lemma~\ref{lem:sumgen}) to label each point by the last added item and a pointer to the previous point, so that the sets generating the points $p$ and $q$ can be reconstructed in time $O(n)$. We omit the details.
\end{proof}

\subsection{Preparations on Rounding}
\label{sec:preprounding}

Here we give two lemmas that bound the difference of the ratio $R(X,Y)$ of two sets $X,Y$ and the ratio $R(\tX,\tY)$ of sets $\tX,\tY$ that are the outcome of rounding $X$ and $Y$. We later use these lemmas to argue correctness of algorithms that round items.

\begin{lemma}[Rounding I] \label{lem:roundingone}
  Let $I \subset \RRp$ of size $n$ and let $\delta > 0$. Round up each item in~$I$ to a multiple of $\delta$ and call the resulting (multi-)set $\widetilde I$. 
  Consider disjoint $\tX,\tY \subseteq \tI$ with $R(\tX,\tY) = 1$.
  Further assume that $\delta \le \frac \eps {2n} \max(\tX \cup \tY)$. Let $X,Y \subseteq I$ be the sets corresponding to $\tX,\tY$ before rounding. Then we have $R(X,Y) \le 1+\eps$.
\end{lemma}
\begin{proof}
  Note that we have $\Sum(X) \le \Sum(\tX) \le \Sum(X) + \delta n$ and $\Sum(Y) \le \Sum(\tY) \le \Sum(Y) + \delta n$. This yields
  \begin{align*}
    R(X,Y) &= \max\left\{ \frac{\Sum(X)}{\Sum(Y)}, \frac{\Sum(Y)}{\Sum(X)} \right\} 
    \le \max\left\{ \frac{\Sum(\tX)}{\Sum(\tY) - \delta n}, \frac{\Sum(\tY)}{\Sum(\tX) - \delta n} \right\}. 
  \end{align*}
  Since $\delta \le \frac \eps {2n} \max(\tX \cup \tY)$ and we have $\max(\tX \cup \tY) \le \Sum(\tX) = \Sum(\tY)$, we obtain
  \begin{align*}
    R(X,Y) 
    &\le \max\left\{ \frac{\Sum(\tX)}{(1-\eps/2) \Sum(\tY)}, \frac{\Sum(\tY)}{(1-\eps/2) \Sum(\tX)} \right\} = \frac 1{1-\eps/2} \le 1+\eps,
  \end{align*}
  where we used $1/(1-\eps/2) \le 1+\eps$ for $\eps \in (0,1)$.
\end{proof}

\begin{lemma}[Rounding II] \label{lem:roundingtwo}
  Let $I \subset \RRp$ of size $n$ with $\OPTL(I) \le 2$ and let $0 < \delta \le \frac \eps{9n} \max(I)$. Round up each item in~$I$ to a multiple of $\delta$ and call the resulting set $\widetilde I$. 
  Consider disjoint $\tX,\tY \subseteq \tI$ with $\max(\tI) \in \tX \cup \tY$ and $R(\tX,\tY) = \OPTL(\tI)$. Let $X,Y \subseteq I$ be the sets corresponding to $\tX,\tY$ before rounding. Then we have
  \[ R(X,Y) \le (1+\eps)\cdot \OPTL(I). \]
\end{lemma}
\begin{proof}
  We start by proving $\OPTL(\tI) \le (1+\tfrac \eps 9) \OPTL(I)$. To this end, consider an optimal solution for $\OPTL(I)$, i.e., disjoint $X',Y' \subseteq I$ with $\max(I) \in X' \cup Y'$ and $R(X',Y') = \OPTL(I)$. 
  By symmetry we can assume $\Sum(X') \ge \Sum(Y')$. Since $\max(I) \in X' \cup Y'$ we have $\Sum(X') \ge \max(I)$, which implies $\delta n \le \frac \eps 9 \Sum(X')$.
  Let $\tX',\tY' \subseteq \tI$ be the sets corresponding to $X',Y'$ after rounding. 
  Note that $\Sum(X') \le \Sum(\tX') \le \Sum(X') + \delta n$ and $\Sum(Y') \le \Sum(\tY') \le \Sum(Y') + \delta n$. 
  This yields
  \begin{align*}
    \OPTL(\tI) \le R(\tX',\tY') = \max\left\{ \frac{\Sum(\widetilde X')}{\Sum(\widetilde Y')}, \frac{\Sum(\widetilde Y')}{\Sum(\widetilde X')} \right\} &\le \max\left\{ \frac{\Sum(X') + \delta n}{\Sum(Y')}, \frac{\Sum(Y') + \delta n}{\Sum(X')} \right\}  \\
    &\le \frac {(1+\frac \eps 9) \Sum(X')}{\Sum(Y')}  \\
    &= (1+\tfrac \eps 9) R(X',Y') = (1+ \tfrac \eps 9) \OPTL(I).
  \end{align*}
  
  Now consider the sets $\tX,\tY,X,Y$ from the lemma statement. 
  By symmetry we can assume $\Sum(X) \ge \Sum(Y)$. Then we have
  \begin{align*}
    \max(I) &= \max(X \cup Y)  & \text{since $\max(\tI) \in \tX \cup \tY$}  \\
    &\le \Sum(X)  & \text{since $\Sum(X) \ge \Sum(Y)$}  \\
    &\le \Sum(\tX)  & \text{by rounding up}  \\
    &\le \OPTL(\tI) \cdot \Sum(\tY)  & \text{as $\tX,\tY$ is an optimal solution for $\OPTL(\tI)$}  \\
    &\le (1 + \tfrac \eps 9) \OPTL(I) \cdot \Sum(\tY)  & \text{as shown in the first paragraph of this proof}  \\
    &\le (1 + \tfrac \eps 9) \cdot 2 \cdot \Sum(\tY)  & \text{by assumption $\OPTL(I) \le 2$}  \\
    &\le 4 \cdot \Sum(\tY).
  \end{align*}
  Together with $\delta \le \frac \eps {9n} \max(I)$ we obtain $\delta n \le \tfrac {4\eps}9 \Sum(\tY)$ and $\delta n \le \tfrac \eps 9 \Sum(\tX)$. 
  Note that we have $\Sum(X) \le \Sum(\tX) \le \Sum(X) + \delta n$ and $\Sum(Y) \le \Sum(\tY) \le \Sum(Y) + \delta n$. 
  We thus obtain
  \begin{align*}
    R(X,Y) = \max\left\{ \frac{\Sum(X)}{\Sum(Y)}, \frac{\Sum(Y)}{\Sum(X)} \right\} 
    &\le \max\left\{ \frac{\Sum(\tX)}{\Sum(\tY) - \delta n}, \frac{\Sum(\tY)}{\Sum(\tX) - \delta n} \right\} \\
    &\le \max\left\{ \frac{\Sum(\tX)}{(1-\tfrac {4\eps}9)\Sum(\tY)}, \frac{\Sum(\tY)}{(1-\tfrac \eps 9)\Sum(\tX)} \right\} \\
    &\le \frac {R(\tX,\tY)}{1-\tfrac {4\eps}9} = \frac{\OPTL(\tI)}{1-\tfrac {4\eps}9} 
    \le \frac {1+\tfrac \eps 9}{1-\tfrac {4\eps}9} \OPTL(I) 
    \le (1+\eps) \OPTL(I),
  \end{align*}
  where in the last step we used $(1+\tfrac \eps 9)/(1-\tfrac {4\eps}9) \le 1+\eps$ for any $\eps \in (0,1)$.
\end{proof}

%
%

\subsection{Putting the Pieces Together}
\label{sec:combination}

Now we are ready to prove a lemma which is almost the main result of Section~\ref{sec:improved}, except that it has the additional requirement $\OPTL(I) \le 2$. We will remove this requirement in Section~\ref{sec:removerequirement}.

\begin{lemma} \label{lem:improvedalgo}
  We can solve an \SSRL instance $(I,\eps)$ in time $\Os(1/\eps^{0.93856})$, assuming $\OPTL(I) \le 2$.
\end{lemma}
\begin{proof}
To achieve this result we combine two algorithms.

\paragraph{Algorithm 1.} This algorithm works as follows:

\begin{enumerate}
\item For a parameter $\tau \in \mathbb{N}$ to be chosen later, partition $I$ into its largest $\tau$ numbers $T$ (top) and the remaining numbers $B = I \setminus T$ (bottom). 

\item Round up all numbers in $I = B \cup T$ to multiples of $\alpha := \frac \eps {9n} \max(I)$, obtaining $\ttI = \ttB \cup \ttT$. Apply Lemma~\ref{lem:maingeometric} to compute disjoint $\ttX,\ttY \subseteq \ttI$ with $\max(\ttI) \in \ttX \cup \ttY$ realizing $\OPTL(\ttI)$. By Lemma~\ref{lem:roundingtwo}, the sets $X,Y \subseteq I$ corresponding to $\ttX,\ttY$ before rounding have ratio $R(X,Y) \le (1+\eps) \OPTL(I)$, so we are done.
\end{enumerate}

For analyzing the running time we introduce the parameter $\psi := \max(T)/\min(T)$. 
Note that Step 3 is dominated by the call to Lemma~\ref{lem:maingeometric}, which runs in time $\Os(S)$ where $S = |P(\ttB)| + |P_L(\ttT)|$ (recall Definition~\ref{def:pointsets} for $P(.)$ and $P_L(.)$). We bound $|P_L(\ttT)| \le 3^\tau$ (since every element of $T$ can be in $X$ or in~$Y$ or in neither of them). 
For $|P(\ttB)|$, we note that the $x$-coordinates of points in $P(\ttB)$ are multiples of $\alpha = \frac \eps {9n} \max(I) = \frac \eps {9n} \max(T)$ of absolute value at most $n \cdot \max(\ttB) \le n \cdot \min(\ttT) \le n \cdot (\min(T) + \alpha)$. Hence, the number of different $x$-coordinates is at most $2 n \cdot (\min(T) + \alpha) / \alpha + 1 = \Os(1/(\eps \psi) + 1)$. Since each point in $P(\ttB)$ has a unique $x$-coordinate, we obtain $|P(\ttB)| = \Os(1/(\eps \psi) + 1)$. The total running time thus is $\Os(1/(\eps \psi) + 3^\tau)$. This finishes the analysis of Algorithm 1.

\paragraph{Algorithm 2.} We now extend the previous algorithm to also make use of Lemma~\ref{lem:mainpigeon}. Algorithm~2 works as follows:

\begin{enumerate}
\item Partition $I$ into its largest $\tau$ numbers $T$ (top) and the remaining numbers $B = I \setminus T$ (bottom). 

\item Round up all numbers in $I = B \cup T$ to multiples of $\beta := \frac \eps {2n} \min(T)$, obtaining $\tI = \tB \cup \tT$. 
Check whether $|\cS(\tB)| > n \max(\tI)/(\beta 2^\tau) + 1$ in time $\Os(\max(\tI)/(\beta 2^\tau) + 1)$ by running the algorithm from Lemma~\ref{lem:computecS}. If this check succeeds, then Lemma~\ref{lem:mainpigeon} is applicable. In this case, we run the algorithm from Lemma~\ref{lem:mainpigeon} on $\tI = \tB \cup \tT$, which is guaranteed to compute disjoint $\tX,\tY \subseteq \tI$ with $(\tX \cup \tY) \cap \tT \ne \emptyset$ and $R(\tX,\tY) = 1$. Now Lemma~\ref{lem:roundingone} is applicable to $\tX,\tY$, since $\beta = \frac \eps {2n} \min(T) \le \frac \eps {2n} \min(\tT) \le \frac \eps {2n} \max(\tX \cup \tY)$. By Lemma~\ref{lem:roundingone}, the sets $X,Y \subseteq I$ corresponding to $\tX,\tY$ before rounding have ratio $R(X,Y) \le 1+\eps$. If we label each number in $\tI$ by its original number in $I$ we can reconstruct $X,Y$ from $\tX,\tY$ in time $\Os(1)$, so we are done.

Hence, from now on we can assume $|\cS(\tB)| \le n \max(\tI)/(\beta 2^\tau) + 1$, so that Lemma~\ref{lem:mainpigeon} is not applicable.

\item Round up all numbers in $I = B \cup T$ to multiples of $\alpha := \frac \eps {9n} \max(T)$, obtaining $\ttI = \ttB \cup \ttT$. Apply Lemma~\ref{lem:maingeometric} to compute disjoint $\ttX,\ttY \subseteq \ttI$ with $\max(\ttI) \in \ttX \cup \ttY$ realizing $\OPTL(\ttI)$. By Lemma~\ref{lem:roundingtwo}, the sets $X,Y \subseteq I$ corresponding to $\ttX,\ttY$ before rounding have ratio $R(X,Y) \le (1+\eps) \OPTL(I)$, so we are done.
\end{enumerate}

To analyze the running time of Algorithm 2, note that
\begin{align}
  \max(\tI)/\beta &\le 1 + \max(I)/\beta  & \text{since we rounded up to multiples of $\beta$}  \notag\\
  & = 1 + \max(T) \cdot n /(\eps \min(T))  & \text{since $T$ contains the largest item of $I$} \notag\\
  & = \Os(\psi/\eps)  & \text{by definition of $\psi$} \label{eq:woashja}
\end{align}
Hence, running Lemma~\ref{lem:mainpigeon} in step 2 takes time $\Os(\max(\tI)^2/(\beta^2 4^\tau) + 3^\tau) = \Os(\psi^2/(\eps^2 4^\tau) + 3^\tau)$. The remaining running time of step 2 is dominated by this term.

As in Algorithm 1, Step 3 is dominated by the call to Lemma~\ref{lem:maingeometric}, which runs in time $\Os(S)$ where $S = |P(\ttB)| + |P_L(\ttT)|$. We again bound $|P_L(\ttT)| \le 3^\tau$. For $|P(\ttB)|$ we now use a different bound.
We observe that each $x$-coordinate of a point in $P(\ttB)$ can be written as the difference of two numbers in $\cS(\ttB)$. We can thus bound $|P(\ttB)| \le |\cS(\ttB)|^2$. As we will show in Claim~\ref{cla:roundsubsetsums} below, we have $|\cS(\ttB)| = \Os(|\cS(\tB)|)$. Recall that after step 2 we can assume $|\cS(\tB)| = \Os(\max(\tI)/(\beta 2^\tau) + 1)$, which we can simplify to $|\cS(\tB)| = \Os(\psi/(\eps 2^\tau) + 1)$ by using inequality~(\ref{eq:woashja}). Combining these bounds yields $|P(\ttB)| = \Os(\psi^2/(\eps^2 4^\tau) + 1)$. 
The total running time of Algorithm 2 is thus $\Os(\psi^2/(\eps^2 4^\tau) + 3^\tau)$. 

\begin{claim} \label{cla:roundsubsetsums}
  We have $|\cS(\ttB)| = \Os(|\cS(\tB)|)$.
\end{claim}
\begin{proof}
  Let $X \subseteq B$ and consider the corresponding subsets $\tX \subseteq \tB$ and $\ttX \subseteq \ttB$.
  Note that $\Sum(X) \le \Sum(\tX) < \Sum(X) + |X| \cdot \beta \le \Sum(X) + n \beta$, and similarly $\Sum(X) \le \Sum(\ttX) < \Sum(X) + n \alpha$. Thus, their difference is bounded by $|\Sum(\tX) - \Sum(\ttX)| < n \alpha$.
  
  Now round up all numbers in $\cS(\tB)$ to multiples of $\alpha$ and remove duplicates to obtain a set $S$. Observe that $|S| \le |\cS(\tB)|$. The number $s \in S$ corresponding to $\Sum(\tX) \in \cS(\tB)$ satisfies $\Sum(\tX) \le s < \Sum(\tX) + \alpha$, and thus by triangle inequality $|s - \Sum(\ttX)| \le |s - \Sum(\tX)| + |\Sum(\tX) - \Sum(\ttX)| < (n+1) \alpha$.
  Note that there are at most $2n+1$ multiples of $\alpha$ within distance less than $(n+1) \alpha$ of $s$. It follows that $\cS(\ttB)$ can be covered by $2n+1$ translates of $S$, more precisely we have 
  $$ \cS(\ttB) \subseteq S + \{-n \alpha,-(n+1) \alpha,\ldots,(n-1)\alpha,n \alpha\}, $$ 
  where we use the sumset notation $A + B = \{a+b \mid a\in A, b \in B\}$. Since $|A+B| \le |A| \cdot |B|$, we obtain $|\cS(\ttB)| \le (2n+1) |S| \le (2n+1) |\cS(\tB)| = \Os(|\cS(\tB)|)$.
\end{proof}

\medskip
Since Algorithm 1 runs in time $\Os(1/(\eps \psi) + 3^\tau)$ and Algorithm 2 runs in time $\Os(\psi^2/(\eps^2 4^\tau) + 3^\tau)$, by running the better of Algorithm 1 and Algorithm 2 we obtain an algorithm running in time 
$$ \Os(\min\{ 1/(\eps \psi), \psi^2/(\eps^2 4^\tau) \} + 3^\tau). $$
Note that both terms in the minimum coincide when $\psi = \eps^{1/3} 4^{\tau/3}$. Thus, by running Algorithm~1 when $\psi \ge \eps^{1/3} 4^{\tau/3}$ and Algorithm 2 when $\psi < \eps^{1/3} 4^{\tau/3}$ we obtain an algorithm running in time $\Os(1/(\eps^{4/3} 4^{\tau/3}) + 3^\tau)$. 
Finally, we determine the optimal choice of $\tau$ by equating the two summands, obtaining $\tau = 4/3 \cdot \log_{3 \cdot 4^{1/3}}(1/\eps)$ (rounded to the closest integer). With this choice of $\tau$, the running time becomes 
$$ \Os( 3^{4/3 \cdot \log_{3 \cdot 4^{1/3}}(1/\eps)} ) = \Os( (1/\eps)^{4/3 \log(3) / \log(3 \cdot 4^{1/3})} ). $$
We numerically evaluate the exponent to be $\approx 0.93855745...$. In particular, our running time is $\Os(1/\eps^{0.93856})$. 
This finishes the proof of Lemma~\ref{lem:improvedalgo}.
\end{proof}

\subsection{The Case of Large Optimal Value}
\label{sec:removerequirement}

Finally, we remove the assumption $\OPTL(I) \le 2$. 
To this end, we present a simple $\sqrt{2}$-approximation algorithm, which is conceptually somewhat simpler than the $(1 + \sqrt{5})/2$-approximation algorithm by Woeginger and Yu~\cite{WoegingerY92}. Importantly, we prove that our approximation algorithm computes an optimal solution if the optimal value is large, i.e., $\OPTL(I) \ge \sqrt{2}$.

\begin{lemma}[$\sqrt{2}$-approximation] \label{lem:algo1}
  Given $I \subset \RRp$ of size $n$, in time $O(n)$ we can compute disjoint $X,Y \subseteq I$ with $\max(I) \in X \cup Y$ and $R(X,Y) \le \max\{\sqrt{2}, \OPTL(I)\}$.
\end{lemma}

\begin{proof}
  As usual, we consider $I$ to be an array $I\arr{1}{n}$ and assume $I[1] \le \ldots \le I[n]$, so $\max(I) = I[n]$. 
  Let $s := \Sum(I\arr{1}{n-1})$. If $s \le I[n]$, then we claim that the optimal solution for $\OPTL(I)$ is $X = \{I[n]\}$ and $Y = I\arr{1}{n-1}$, resulting in $\OPTL(I) = R(X,Y) = \Sum(X)/\Sum(Y) = I[n]/s$. Indeed, one of $X$ or $Y$ must contain $I[n]$, so suppose $I[n] \in X$. Then $\Sum(X) \ge I[n]$ and $\Sum(Y) \le s$, so $R(X,Y) \ge \Sum(X)/\Sum(Y) \ge I[n]/s$. Thus, our chosen sets $X,Y$ are optimal. 
  It follows that in case $s \le I[n]$ we can easily compute an optimal solution.
  
  So consider the case $s > I[n]$. Then we pick the largest $1 \le i \le n-1$ such that $\Sum(I\arr{i}{n-1}) > I[n]$. Note that such a value $i$ exists since $\Sum(I\arr{1}{n-1}) = s > I[n]$. Moreover, we have $i < n-1$, since $\Sum(I\arr{n-1}{n-1}) = I[n-1] \le I[n]$ by sortedness. 
  We now set $X := \{I[n]\}$, $Y := I\arr{i}{n-1}$, and $Y' := I\arr{i+1}{n-1}$. Note that $\Sum(Y) > \Sum(X) \ge \Sum(Y')$. Moreover, by sortedness we have $I[i] \le I[i+1] \le \Sum(Y')$ and thus $\Sum(Y) = \Sum(Y') + I[i] \le 2\Sum(Y')$. Hence, the better of the solutions $(X,Y)$ and $(X,Y')$ has ratio
  \[ \min\left\{R(X,Y), R(X,Y')\right\} = \min\left\{ \frac{\Sum(Y)}{\Sum(X)}, \frac{\Sum(X)}{\Sum(Y')} \right\} \le \sqrt{\frac{\Sum(Y)}{\Sum(Y')}} \le \sqrt{2}. \]
  
  In both cases we obtain a solution of ratio at most $\max\{\sqrt{2}, \OPTL(I)\}$. Implementing this algorithm in time $O(n)$ is straightforward, for pseudocode see  Algorithm~\ref{alg:algo1}.
  %
\end{proof}

\begin{algorithm}[!t]\caption{Computes a solution of ratio at most $\max\{\sqrt{2}, \OPTL(I)\}$, see Lemma~\ref{lem:algo1}.}\label{alg:algo1}
\begin{algorithmic}[1]
\Procedure{\textsc{Algorithm1}}{$I\arr{1}{n}$} 
 \Comment{given $I\arr{1}{n}$, assume $0 < I[1] \le \ldots \le I[n]$}
 \If {$I[1]+\ldots+I[n-1] \le I[n]$}
   \State \Return $(\{I[n]\}, I\arr{1}{n-1})$ 
 \Else
   \State Perform linear scan for largest $i$ such that $I[i]+I[i+1]+\ldots+I[n-1] > I[n]$ 
   \State $X := \{I[n]\}$, $Y := I\arr{i}{n-1}$, $Y' := I\arr{i+1}{n-1}$
   \If{$R(X,Y) \le R(X,Y')$} \Return $(X,Y)$
   \Else{} \Return $(X,Y')$
   \EndIf
 \EndIf
\EndProcedure
\end{algorithmic}
\end{algorithm}

The above lemma yields a solution of ratio at most $\OPTL(I)$ (and thus solves the \SSRL instance~$(I,\eps)$) if $\OPTL(I) \ge \sqrt{2}$. Moreover, it runs in time $\Os(1)$. After running this algorithm, we can assume $\OPTL(I) < \sqrt{2}$, as otherwise we already solved the \SSRL instance. Therefore, the requirement $\OPTL(I) \le 2$ of Lemma~\ref{lem:improvedalgo} is satisfied, meaning that we can solve the \SSRL instance~$(I,\eps)$ in time $\Os(1/\eps^{0.93856})$. Together, we thus obtain the following theorem.

\begin{theorem} \label{thm:improved}
  We can solve a given \SSRL instance $(I,\eps)$ in time $\Os(1/\eps^{0.93856})$.
\end{theorem}

As argued before, by combining the Reduction Lemma~\ref{lem:reduction} with Theorem~\ref{thm:improved} we obtain an algorithm for \SSR with running time $O(n/\eps^{0.93856} \polylog(1/\eps)) = O(n/\eps^{0.9386})$, thus proving Theorem~\ref{thm:main}.

\bibliography{main}

\end{document}